\documentclass[11pt]{article}
\usepackage{amsmath}
\usepackage{amssymb}
\usepackage{listings}
\usepackage{color}
\usepackage{algorithm}
\usepackage[noend]{algorithmic}
\usepackage{amssymb}
\usepackage{mathtools}
\usepackage{xspace}
\usepackage{bm}
\DeclarePairedDelimiter{\ceil}{\lceil}{\rceil}
\DeclarePairedDelimiter{\floor}{\lfloor}{\rfloor}

\usepackage{definitions}

\newtheorem{problem}[lemma]{Problem}

\usepackage{color-edits}
\addauthor[David]{dk}{red}
\addauthor[Mark]{mk}{blue}

\newcommand{\vc}[1]{\bm{#1}\xspace}

\newcommand{\AP}[1][]{\ensuremath{\mathcal{X}_{#1}}\xspace}

\newcommand{\SCH}{\ensuremath{\ell}\xspace}
\newcommand{\Sch}[1]{\ensuremath{\ell(#1)}\xspace}

\newcommand{\VSch}[2][]{\ensuremath{%
\ifthenelse{\equal{#1}{}}{\ell^{-1}(#2)}{\ell^{-1}_{#1}(#2)}}\xspace}
\newcommand{\PCost}[3][]{\ensuremath{%
\ifthenelse{\equal{#1}{}}{c_{#2}(#3)}{c^{(#1)}_{#2}(#3)}}\xspace}
\newcommand{\MCost}[2][]{\ensuremath{%
\ifthenelse{\equal{#1}{}}{C(#2)}{C^{(#1)}(#2)}}\xspace}

\newcommand{\OPTVALUE}{\ensuremath{C^*}\xspace}
\newcommand{\OptTreeCost}[1][]{\ensuremath{C^{\text{MST}}_{#1}}\xspace}
\newcommand{\OPTMINMAX}{\ensuremath{\text{OPT}_{\text{Min-Max}}}\xspace}
\newcommand{\OPTTOUR}{\ensuremath{\ell^*}\xspace}
\newcommand{\OPTTSP}{\ensuremath{\text{OPT}_{\text{TSP}}}\xspace}
\newcommand{\MST}{\ensuremath{\text{MST}}\xspace}

\begin{document}

\pagestyle{empty}

\title{A Class of Weighted TSPs with Applications}
\author{David Kempe \qquad Mark Klein\\ University of Southern California}

\maketitle

\begin{abstract}

Motivated among others by applications to the prevention of poaching
or burglaries, we define a class of weighted Traveling Salesman
Problems on metric spaces.
The goal is to output an infinite (though typically periodic) tour
that visits the $n$ points repeatedly,
such that no point goes unvisited for ``too long.''
We consider two objective functions for each point $x$.
The \emph{maximum} objective is simply the maximum duration of any
absence from $x$,
while the \emph{quadratic} objective is the normalized sum of squares
of absence lengths from $x$.
The overall minimization objective is the weighted maximum of the
individual points' objectives.
When a point has weight $w_x$, the absences under an optimal tour
should be roughly a $1/w_x$ fraction of the absences from points of
weight 1.
Thus, the objective naturally encourages visiting high-weight points
more frequently, and at roughly evenly spaced intervals.

We give a polynomial-time combinatorial algorithm whose output is
simultaneously an $O(\log n)$ approximation under both objectives.
We prove that up to constant factors, approximation guarantees for the
quadratic objective directly imply the same guarantees for a natural
security patrol game defined in recent work.

\end{abstract}

\section{Introduction} \label{sec:introduction}

In the classic metric Traveling Salesman Problem (TSP),
we are given $n$ points in a metric space,
and the goal is to compute a tour that visits each point (at least) once,
returns to the starting point,
and (approximately) minimizes the total distance traveled.
The name is derived from the story of a traveling
salesman who needs to sell his product in each of $n$ cities,
and wants to return home as quickly as possible.%
\footnote{Beyond the namesake story, the TSP has found many
important real-world applications,
including vehicle routing, wiring of computers or drilling of holes in
chip boards, and job sequencing \cite{LLRS:tsp-survey-book}.}
  
In reality, the need for the salesman's product or service will
not typically arise once at time 0 and be forever met with the
salesman's visit.
Instead, demand for the product will arise over time in the
population, and the salesman will need to return to the same cities
repeatedly to serve the demand that has accrued since his most
recent visit. 
Naturally, demand in larger cities accrues more quickly, so larger
cities should be visited more frequently.
Unserved demand leads to disutility, and the salesman's objective
function in choosing an (infinite) tour is to minimize the overall
disutility.


This view motivates the following class of
\emph{weighted Traveling Salesman Problems}
(defined more formally in Section~\ref{sec:preliminaries}).
Each of the $n$ points in the metric space has a weight $w_x \geq 0$.
For each point $x$, a (infinite) visit schedule induces a distribution
of times (total distance traveled) between consecutive visits to $x$
(which we call \emph{absence lengths}). 
For the two versions of the problem we study,
the cost $c_x$ of a point $x$ is either
(1) the maximum absence length from $x$
(we call this version the \emph{weighted Max-TSP}),
or
(2) the expected absence length at a time $t$ that is chosen uniformly
randomly from the schedule%
\footnote{This interpretation has meaning only when the schedule is
  periodic. We define a more general notion in
  Section~\ref{sec:preliminaries}.},
which equals the normalized sum of squares of absence lengths
(we call this version the \emph{weighted Quadratic TSP}).
The goal is then to find an infinite schedule minimizing
$\max_x w_x c_x$.
The objective function encourages the salesman to visit
high-weight points more frequently;
and not only more frequently, but ``roughly evenly spaced.''

In terms of the traveling salesman,
this class of problems is motivated as follows.
After the salesman has committed to his tour,
a city $x$ is chosen adversarially, and develops a need for the
salesman's product/service.
The cost incurred is the duration until this need is met,
times the population of the city incurring the need.
If the time is also chosen adversarially,
we exactly obtain the weighted Max-TSP objective;
if the time is random, then the expected cost is (half) the weighted
Quadratic TSP objective.

An alternative objective
--- discussed briefly in Sections~\ref{sec:related-work} and
\ref{sec:conclusions} 
--- is to minimize the sum $\sum_x w_x c_x$ instead of the maximum.
These objective functions are motivated when the city is not chosen
adversarially, but rather, a random person from the population
develops a need for the product.

\subsection{A Second Application: Security Games}
\label{sec:security-games-intro}

A second important motivation arises from a class of security
games, especially in the prevention of poaching, illegal
harvesting, or burglary. 

In general, in security games \cite{tambe:security-games-book,balcan:blum:haghtalab:procaccia,xu:mysteries},
a defender needs to use limited resources to protect a set of $n$ targets
of non-uniform value $w_x$ from attack by a rational attacker.
The attacker can observe the (typically mixed) strategy chosen by the
defender, but usually not the defender's coin flips.
A wide literature has analyzed different variants of this problem,
with different combinatorial constraints on the actions of either or
both players; see \cite{tambe:security-games-book,xu:mysteries} for a
comprehensive survey and a general optimization treatment.

Motivated among others by anti-poaching
\cite{FNPLCASSTL:PAWS,fang:stone:tambe:green} 
and anti-burglary efforts,
a recent paper \cite{QuasiRegular} proposed a variant
in which the defender commits to a distribution over (infinite) tours
of the $n$ targets; in \cite{QuasiRegular}, the metric is assumed to
be uniform, i.e., all targets are at distance 1 from each other.
The attacker performs a single attack,
choosing a target $x$ to attack, as well as an attack duration $t$. 
The attacker's reward if not caught is $w_x \cdot t$, and if caught,
is 0.
In Section~\ref{sec:additional-results},
we show a reduction from this problem to the weighted quadratic TSP
objective discussed above; 
the reduction only loses a constant factor in the approximation.

If the attacker \emph{can} observe the defender's coin flips,
then the defender may as well not randomize.
A rational attacker will attack a target $x$ immediately after the
defender has left it for the longest possible interval,
and stay until right before the defender returns to it.
Thus, the attacker's reward for the best possible attack is exactly
the weighted Max-TSP objective.

\subsection{Our Results}
Our main result is an efficient (combinatorial) $O(\log n)$
approximation algorithm for both the weighted Max-TSP and the
weighted Quadratic TSP.
We state the result informally here; a formal statement and proof are
given in Section~\ref{sec:main-result}.

\begin{theorem}[Main Theorem (stated informally)]
  \label{thm:main-informal}
There is a polynomial-time approximation algorithm which returns a
solution that is simultaneously an $O(\log n)$ approximation for both
the weighted Max-TSP and weighted Quadratic TSP objectives.
\end{theorem}

As a corollary, we obtain $O(\log n)$ approximation algorithms for the
security game scheduling problems discussed in
Section~\ref{sec:security-games-intro}.

In our algorithm, we first round weights to powers $2^{-i}$.
The points of weight $2^{-i}$ are partitioned (approximately
optimally) into $2^i$ spanning trees $T^i_j$, such that no spanning
tree is too large.
Each spanning tree is shortcut into a tour, and those tours are
carefully sequenced, such that tours for points of higher weight are
repeated correspondingly more frequently.

In Section~\ref{sec:preliminaries}, we give
an approximation-preserving reduction from the
standard metric TSP to the weighted Max-TSP,
implying APX-hardness of the latter by \cite{papadimitriou:yannakakis:TSP}.
However, there is an obvious gap between the constant lower bound for
approximability and the $O(\log n)$ upper bound we obtain.
The question of whether the approximation guarantee can be improved to
a constant is a natural question for future work.

\subsection{Related Work} \label{sec:related-work}
The motivation for our problem is similar to that of the well-known
Minimum Latency Problem
\cite{arora:karakostas:latency,BCCPRS:minimum-latency,goemans:kleinberg:latency}
(also called the Traveling Repairman Problem).
In the Minimum Latency Problem, the objective is to minimize the \emph{sum}
of first arrival times to each point in a given (finite) metric space;
a natural motivation is that a task is to be performed at each point,
and the goal is to serve each point as quickly as possible on average.
The high-level approach of most approximation algorithms
for this problem is to consider trees of exponentially increasing
total cost, each spanning (approximately) as many points as possible,
then to execute the sequence of Eulerian tours of those trees.
The arrival time at a point that first appears in tree $i$ is
dominated by the cost of the \Kth{i} tree,
and the fact that the point was in no earlier tree serves as a
certificate that not many more points could have been visited faster.
A constant-factor approximation algorithm for the minimum-cost tree
spanning $k$ vertices \cite{garg:k-MST} therefore yields a
constant-factor approximation for the Minimum Latency Problem.

The fact that each point only needs to be visited once sets the Minimum
Latency Problem apart from our problem.
Once a point has been visited, it can be ignored, whereas in our
problem, a high-weight point must be visited repeatedly.
Following the ``exponentially increasing trees'' approach outlined
above, the final tour (of largest weight) visits all nodes,
and will therefore result in a long absence from each point.
In our weighted TSP, if one or a few points have high weight,
this can cause the objective function to be a factor of $\Omega(n)$
from optimal (e.g., in a star in which the center node has weight $n$).
As a result, it appears unlikely that algorithms based on this
approach can yield good approximation guarantees for our problem.
Indeed, our algorithm in Section~\ref{sec:main-result} is based
instead on a decomposition into multiple spanning trees, each of which
is sufficiently small and gives rise to a tour; these tours are then
carefully sequenced to ensure that high-weight points are visited
regularly enough.
    
As further evidence of the difference between the problems,
notice that in the Minimum Latency Problem, a point $x$ of weight
$w_x$ can simply be replaced by $w_x$ points of weight 1;
i.e., the unweighted problem is fully general.
By contrast, for the Weighted Max-TSP objective, the objective becomes
exactly the standard TSP objective when all weights are 1,
and for the Weighted Quadratic TSP with uniform weights,
the optimum TSP solution is within a constant factor of optimal
(see Section~\ref{sec:preliminaries}).

Closely related to the Minimum Latency Problem are the broad classes of
problems in vehicle routing and maintenance/machine scheduling. 
In vehicle routing, one or more vehicles must visit locations to serve demands.
Different variants have different constraints on the number or
initial locations of vehicles, time windows during which the cities
must be visited, etc.
The objectives are typically to minimize the total length/distance of
a tour, or maximize the demand served within a given amount of time
\cite{lenstra:kan}.
Constant-factor approximations are known for many variants of the
problem, but as with the Minimum Latency Problem, the techniques do
not appear to carry over to our problem, because each city must only
be visited once.

One recent approach that has improved approximation guarantees
for (multi-vehicle) Minimum Latency problems is the use of
time-indexed configuration LPs \cite{chakrabarty:swamy,post:swamy}.
Here, for each time $t$ and path $P$ (starting from an appropriate
origin node), a variable $z_{t,P}$ captures whether $P$ is the path
taken up to time $t$ (possibly with another index for the vehicle).
In the fractional relaxation, one obtains a ``distribution'' of
arrival times at nodes $v$; it can be shown that sufficiently strong
versions of these configuration LPs (which capture joint paths for all
vehicles) have small integrality gaps. The analysis exploits more
sophisticated variants of the arguments outlined above for trees of
expanding total cost.

While configuration LPs are an approach that explicitly models time in
the routing problem,
it is not clear that it can be extended to our weighted TSPs.
The configuration LPs do not (and do not need to) enforce any
consistency between paths up to time $t$ and paths up to time $t+1$.
It is not at all clear how one would express the time between
\emph{consecutive} visits to a point in terms of the LP variables of
configuration LPs, though it is an intriguing possibility that another
type of configuration LP could yield improved approximation guarantees.

The variant of vehicle routing most similar to our problem is the
vehicle routing problem with time windows, or TSP with time windows
(TSPTW).
Here, each city has a single time window, and in order to serve its
demand, the city must be visited during this window.
The objective is usually phrased as maximizing the number/weight of
cities whose time window is met by the vehicle(s).
The approximability of this problem has been well studied
for different metric spaces.
For general metrics, \cite{bockenhauer:hromkovic:kneis:kupke} showed
that TSPTW does not admit a constant-factor approximation unless
$\text{P} = \text{NP}$.
The best known approximation for general metric spaces is an
$O(\log^2(n))$ approximation by \cite{bansal:blum:chawla:meyerson}.
They first give an $O(\log(n))$ approximation for the deadline TSP
problem, and then shows that any $\alpha$-approximation for deadline
TSP leads to an $\alpha^2$ approximation for TSPTW.
In addition, \cite{bansal:blum:chawla:meyerson} give a bicriteria
approximation scheme that allows for a parameter $\epsilon$ indicating
by what factor deadlines may be exceeded.

The bicriteria approximation scheme of
\cite{bansal:blum:chawla:meyerson} could possibly be leveraged to
address our weighted TSPs.
After guessing the objective value, one obtains absence times for each
point, and can --- for each point --- partition the timeline into
intervals of approximately that length, each interval corresponding to
a new ``city'' in the TSPTW problem.
It might be possible to search for a suitable $\epsilon$ and use the
algorithm from \cite{bansal:blum:chawla:meyerson} as a certifier of
feasibility or infeasibility of an objective value.
However, it is not clear whether or how the approximation guarantees
from \cite{bansal:blum:chawla:meyerson} will translate into
approximation guarantees for our weighted TSPs, especially because
our problem effectively requires an infinite number of time windows
and approximation guarantees of \cite{bansal:blum:chawla:meyerson} are
given in terms of the number of time windows (since there is one time
window for each city).


An objective function more closely related to ours is minimized in
machine repair scheduling problems
\cite{bar-noy:bhatia:naor:schieber,anily:glass:hassin}.
Here, each day, a limited number of machines can be repaired,
and the goal is to choose a schedule that minimizes the expected value
of a randomly selected machine's weight,
scaled linearly by  the number of days since it was last repaired.
This of course necessitates that each machine be visited repeatedly.
However, any machine can be repaired in any time step, which
implicitly corresponds to the uniform metric also studied in
\cite{QuasiRegular}.
For the uniform metric, constant approximation can be easily obtained,
but it does not appear that these can be extended easily to general
metrics.


Much recent work has focused on improving the approximation guarantee
for the metric TSP below the factor $\frac{3}{2}$ which one obtains
from Christofides's Algorithm.
This approach is orthogonal to ours, as we currently do not even know
of a constant-factor approximation for our class of problems.
The techniques used in this line of work are typically based on
rounding the Held-Karp relaxation.
Given that our problems are much more sensitive to the specific
sequencing of the visits, it appears unlikely that any modification of
Held-Karp (or other standard LP rounding approaches) would be useful.

As we mentioned above, our work is directly motivated by recent work
\cite{QuasiRegular} on security games for poaching or burglary
prevention. \cite{QuasiRegular} studied only the uniform metric,
and gave \emph{optimal} algorithms for a natural class of objectives.
Because we show that the objective function of \cite{QuasiRegular} is
within a constant factor of the Weighted Quadratic TSP objective, 
our result gives an approximation guarantee for the objective of
\cite{QuasiRegular} for arbitrary metrics, as shown in
Section~\ref{sec:additional-results}.
A similar motivation of preventing wildlife poaching underlies recent
work \cite{immorlica:kleinberg:bandits} on recharging bandits.
Among others, \cite{immorlica:kleinberg:bandits} define a maximization
version in which the reward at each point is an increasing concave
function of the time since the last visit, and the goal is to find a
tour (approximately) maximizing the total reward collected.
As with \cite{QuasiRegular}, \cite{immorlica:kleinberg:bandits} only
considers the case of uniform distances; obtaining provable
approximation guarantees for non-uniform distances in the maximization
version is an interesting direction for future work.

One special class of security games most closely related to our
objective is that of \emph{adversarial patrolling games}
\cite{basilico:gatti:amigoni:leader,vorobeychik:an:tambe:patrolling}.
Similar to our work, the goal there is to define a tour such that an
attacker who gets to observe the tour will cause as little damage as
possible. The exact objective functions differ from ours,
and the literature focuses on heuristics or exponential-time
algorithms rather than polynomial-time algorithms with provable
approximation guarantees.

\section{Definitions and Preliminary Results} \label{sec:preliminaries}
The metric space consists of a set \AP of $n \geq 3$ points $x \in \AP$
with positive distances $d(x,y) > 0$.
For the purpose of this paper,
a \emph{tour} is a mapping $\SCH : \N \to \AP$.
The interpretation is that the salesman visits the points in the order
$\Sch{1}, \Sch{2}, \ldots$.
Tours will typically visit points infinitely often,
and in principle may skip some points completely or for long periods of time;
however, such tours are far from optimal.

The total distance or time (we use the terms interchangeably)
traveled by the salesman under \SCH from step $t$ to $t'$ is
$d_{\SCH}(t,t') := \sum_{\tau=t}^{t'-1} d(\Sch{\tau}, \Sch{\tau+1})$.
In order to evaluate the quality of a tour,
we are interested in the maximum absence from any point,
and in the normalized sum of squares of absences.
Fix a point $x$, and let \VSch[x]{k} be the time step in which the
\Kth{k} visit to point $x$ occurs under \SCH.
Precisely, $\Sch{\VSch[x]{k}} = x$, and there are exactly $k$ times
$t \leq \VSch[x]{k}$ with $\Sch{t} = x$.
If a schedule visits point $x$ only $k$ times,
then $\VSch[x]{k'} = \infty$ for $k' > k$.
For notational convenience, we define $\VSch[x]{0} = 1$.
We then define the following two cost functions:

\begin{align}
\PCost[\infty]{x}{\SCH}
  & = \max_{k=0,1,2,\ldots} d_{\SCH}(\VSch[x]{k},\VSch[x]{k+1}), \label{eqn:point-cost-infty}
\\
\PCost[2]{x}{\SCH}
  & = \limsup_{k \to \infty}
    \frac{\sum_{k'=0}^{k} (d_{\SCH}(\VSch[x]{k'},\VSch[x]{k'+1}))^2}{%
    d_{\SCH}(1,\VSch[x]{k+1})}. \label{eqn:point-cost-p}
\end{align}  

\begin{remark}
A few remarks are in order about these definitions:
\begin{enumerate}
\item The first objective --- which we call the \emph{weighted Max-TSP
    objective} --- simply measures the longest absence from $x$,
which will be infinite if $x$ is only visited finitely many times.

\item The second objective --- which we call the \emph{weighted Quadratic
  TSP objective} --- intuitively captures the expected length of an
absence interval, if intervals are selected with probability
proportional to their length.
If \SCH is periodic, this is exactly the definition. 

\item The normalization in the definition of \PCost[2]{x}{\SCH} 
is necessary to ensure that the cost is invariant under
concatenating a periodic schedule.
In other words, we want to ensure that the cost is the same
for a finite schedule repeated infinitely,
and the schedule repeated twice repeated infinitely.

\item In the definition of \PCost[\infty]{x}{\SCH},
unless $\PCost[\infty]{x}{\SCH} = \infty$,
the maximum is actually attained.
The reason is the following.
The length of any subtour is the sum of a multiset of pairwise
distances.
Because the distances are lower-bounded by the shortest distance,
there are only finitely many different lengths of tours below any
given bound $B$.
Hence, if the supremum is $B < \infty$, it can only occur as a maximum.
\end{enumerate}
\end{remark}

Each point $x$ has a weight $w_x > 0$.
Without loss of generality (by rescaling),
we assume that the largest weight is 1.
Using these weights,
the overall cost function is the weighted maximum of the individual
points' cost functions.
Formally, it is defined as follows (for $p \in \SET{2, \infty}$):

\begin{align}
\MCost[p]{\SCH}
  & = \max_x w_x \PCost[p]{x}{\SCH}. \label{eqn:max-cost}
\end{align}


\subsection{Preliminary Results}

We begin with an easy approximation-preserving reduction from the
metric TSP to the weighted Max-TSP.
This result establishes APX-hardness using the APX-hardness of the
metric TSP \cite{papadimitriou:yannakakis:TSP},
and the bound will also be important for our analysis in
Section~\ref{sec:main-result}.

\begin{lemma} \label{lem:TSP-to-weighted}
Consider an instance in which all weights are the same:
$w_x = 1$ for all $x$.
Let \OPTTSP be the cost of the optimum TSP tour,
and \OPTVALUE the cost of the optimum tour with respect to the
\MCost[\infty]{\SCH} objective.
Then, $\OPTTSP = \OPTVALUE$.
\end{lemma}

\begin{proof}
First, repeating the optimum TSP tour gives a tour in which each
absence from each point $x$ (and in particular the maximum absence) is
exactly \OPTTSP.
Therefore, $\OPTVALUE \leq \OPTTSP$.

For the converse inequality, let \OPTTOUR be the optimum tour.
Let $x$ be a point maximizing the length of any absence between
consecutive visits to $x$. 
Let $t, t'$ be time steps of consecutive visits to $x$ under \OPTTOUR
achieving this maximum,
and let $D = d_{\OPTTOUR}(t,t')$
be the total time between these consecutive visits.

Between $t$ and $t'$, all points must be visited at least once.
If not, then a point $y$ not visited between these steps would have a
longer absence interval than $x$,
namely, at least from step $t-1$ to $t'+1$.
Because each point is visited at least once,
the subtour of \OPTTOUR from step $t$ until $t'$ is a candidate TSP
tour, showing that $\OPTTSP \leq \OPTVALUE$.
\end{proof}

\begin{lemma} \label{lem:periodic-is-near-OPT}
Under the weighted Quadratic TSP objective, for every tour \SCH,
there is a periodic tour $\SCH'$ such that
$\MCost[2]{\SCH'} \leq 3 \MCost[2]{\SCH}$.
\end{lemma}

\begin{proof}
The claim is trivial if $\MCost[2]{\SCH} = \infty$,
so we focus on the case $\MCost[2]{\SCH} < \infty$.
Let $k_0$ be large enough such that for all $k \geq k_0$ and all
points $x$, we have 
\[\frac{\sum_{k'=0}^{k} (d_{\SCH}(\VSch[x]{k'},\VSch[x]{k'+1}))^2}{%
d_{\SCH}(1,\VSch[x]{k+1})}
\leq \sqrt{3/2} \cdot \PCost[2]{x}{\SCH}.\]
Let $t_0 = \max_x \VSch[x]{k_0+1}$,
and let $t_1$ be large enough such that for all $x$ and all
$k$ with $\VSch[x]{k} \geq t_1$, we have that
$d_{\SCH}(1,\VSch[x]{k+1}) \leq \sqrt{3/2} \cdot d_{\SCH}(1,\VSch[x]{k})$.
Such a $t_1$ exists:
otherwise, there would be arbitrarily large $k$ such that 
$d_{\SCH}(\VSch[x]{k},\VSch[x]{k+1}) > (\sqrt{3/2}-1) \cdot
d_{\SCH}(1,\VSch[x]{k})$;
because this means that a constant fraction of arbitrarily long
subtour lengths is spent on the last absence from $x$,
the quadratic objective for $x$ would diverge to $\infty$.

Let $\hat{t} = \max(t_0, t_1)$,
and let $\SCH'$ be the periodic schedule repeating
$\Sch{1}, \Sch{2}, \ldots, \Sch{\hat{t}}$.

Fix a point $x$, and let $k_x \geq k_0$ be the number of absence
intervals from $x$ in \SCH that are contained in $\SCH'$.
For $k = 1, \ldots, k_x-1$, the absence intervals have lengths
$d_{\SCH} (\VSch[x]{k}, \VSch[x]{k+1})$.
The final absence interval, in which the tour wraps around,
has length at most 
$d_{\SCH} (\VSch[x]{k_x}, \VSch[x]{k_x+1}) + d_{\SCH} (1, \VSch[x]{1})$.
In the quadratic objective, we are interested in the square of this
length, which is at most
\[
  (d_{\SCH} (\VSch[x]{k_x}, \VSch[x]{k_x+1}) + d_{\SCH} (1, \VSch[x]{1}))^2
\; \leq \; 
2 d_{\SCH}^2 (\VSch[x]{k_x}, \VSch[x]{k_x+1}) +
2 d_{\SCH}^2 (1, \VSch[x]{1}).
\]

Thus, the objective function value for $x$ under $\SCH'$ is at most
\begin{eqnarray*}
&& \frac{2 d_{\SCH}^2(1, \VSch[x]{1})
    + 2 d_{\SCH}^2(\VSch[x]{k_x},\VSch[x]{k_{x}+1})
    + \sum_{k=1}^{k_x-1} d_{\SCH}^2(\VSch[x]{k},\VSch[x]{k+1})}{%
d_{\SCH}(1,\VSch[x]{k_x})}
\\ & \leq &
\frac{2\sum_{k=0}^{k_x} d_{\SCH}^2(\VSch[x]{k},\VSch[x]{k+1})}{%
\sqrt{2/3} \cdot d_{\SCH}(1,\VSch[x]{k_x+1})}
\\ & \leq & 3 \PCost[2]{x}{\SCH}.
\end{eqnarray*}
Since this holds for all points $x$, the proof is complete.
\end{proof}

The following lemma compares the optimum TSP cost to the (unweighted)
Quadratic TSP objective. It is central to the proof of our main
theorem.

\begin{lemma} \label{lem:quadratic-vs-TSP}
Consider an instance in which all weights are the same:
$w_x = 1$ for all $x$.
Let \OPTTSP be the cost of the optimum TSP tour,
and \OPTVALUE the cost of the optimum tour with respect to the
\MCost[2]{\SCH} objective.
Then, $\OPTTSP \leq 480 \cdot \OPTVALUE$.
\end{lemma}

\begin{proof}
By Lemma~\ref{lem:periodic-is-near-OPT}, we can focus, at a loss of a
factor of 3, on periodic tours.
Consider a periodic tour \SCH of $\hat{t}$ steps,
with $\Sch{\hat{t}+1} = \Sch{1}$.
Let $q = \OPTTSP/16$.
We consider two cases, based on whether \SCH contains many long edges
or not.
Let $T$ be the set of all ``long'' time steps $t$, i.e.,
time steps $t$ with $d(\Sch{t}, \Sch{t+1}) \geq q$.

\begin{itemize}
\item In the first case, we assume that
$\sum_{t \in T} d(\Sch{t}, \Sch{t+1})
\geq \frac{1}{10} \cdot d_\SCH(1, \hat{t}+1)$.
Consider an arbitrary point $x$.
Because the tour is absent from $x$ for at least each time interval
corresponding to a step $t \in T$,
the numerator of the quadratic TSP objective for $x$
is at least
\[
  \sum_{t \in T} d^2(\Sch{t}, \Sch{t+1})
\; \geq \; |T| \cdot \left( \frac{1}{|T|} \cdot \sum_{t \in T}
  d(\Sch{t}, \Sch{t+1}) \right)^2
\; = \; \frac{\left( \sum_{t \in T} d(\Sch{t}, \Sch{t+1}) \right)^2}{|T|},
\]
where the inequality was using convexity.
The denominator of the quadratic objective is
$d_\SCH(1, \hat{t}+1) \leq 10 \sum_{t \in T} d(\Sch{t}, \Sch{t+1})$,
so the quadratic TSP objective for $x$ is at least
\[
\frac{\sum_{t \in T} d(\Sch{t}, \Sch{t+1})}{10 |T|}
\; \geq \; \frac{q}{10}
\; = \; \frac{1}{160} \cdot \OPTTSP.
\]

\item Otherwise, we will focus exclusively on time steps $t$ for which
$d(\Sch{t}, \Sch{t+1}) < q$.
Produce subtours (paths) of length between $q$ (inclusive) and $2q$
(exclusive) as follows:
\begin{enumerate}
\item Remove from \SCH all edges of length strictly more than $q$.
  This produces maximal subpaths with the property that each such
  subpath is immediately preceded and followed by an edge of length
  strictly greater than $q$ in \SCH. (Except in the case that there
  were no such long edges; this case is trivial.)
\item Discard all subpaths whose total length is strictly less than
  $q$.
\item For each remaining subpath, of length at least $q$, partition it
  greedily into segments of length between $q$ (inclusive) and $2q$
  (exclusive). If as a result of this greedy partition, there is a
  (final) segment of length strictly less than $q$, discard it.
  Notice that such a partitioning is always possible because each
  remaining edge has length at most $q$.
\end{enumerate}

Let $K$ be the number of such subtours that are produced,
and $E_1, E_2, \ldots, E_K$ the corresponding edge sets;
each $E_i$ defines a path (though not necessarily a simple path).
We first show that the subtours together comprise at least
a $4/5$ fraction of the total length of \SCH.
The reason is the following: whenever a subtour is discarded, its
length was strictly less than $q$, and it must have been immediately
followed by an edge of length greater than $q$.
For each long edge, at most one short subtour is discarded due to this
edge, so the total length of discarded edges (long edges and short
tours) is at most twice the total length of long edges, which is at
most $1/10$ the total length of all edges (because of the current case).

We will prove below that at least one point $x$ must be absent from at
least $1/4$ of these subtours $E_i$.
Then, in the numerator of the quadratic TSP objective for this point $x$,
there will be at least $K/4$ terms of at least $q^2$ each.
The total length of all the $E_i$ is less than $2qK$, and the total
length of discarded edges is at most another $qK/2$.
Thus, the denominator is less than $5qK/2$,
and the quadratic TSP objective is at least $q/10 = \OPTTSP/160$.

It remains to prove that some point $x$ is absent from at least $1/4$ of the
subtours.
Assume for contradiction that every point is in more than
$3/4$ of the subtours $E_i$.
Each $E_i$ induces a single connected component on the nodes it
includes, and has total cost at most $2q$.
By Lemma~\ref{lem:steiner-trees} below
(with $\gamma = \frac{1}{4}$ and $C=2q$),
this implies the existence of a tree spanning all points of cost
strictly less than $8 q$,
and thus a TSP tour of total cost strictly less than $16 q$,
by shortcutting the Eulerian tour of the spanning tree.
This is a contradiction to the definition that
$q = \OPTTSP/16$.
\end{itemize}

Taking the maximum of both cases,
and accounting for the factor of 3 we lost by focusing on a periodic schedule,
gives us the claimed bound.
\end{proof}

\begin{lemma} \label{lem:steiner-trees}
Let $G$ be an undirected graph with non-negative edge weights.
Let $E_1, E_2, \ldots, E_K \subseteq E$ be connected edge sets
of cost $c(E_i) < C$,
i.e., each $E_i$ induces a single connected component
(with the remaining vertices being isolated).
Assume that each vertex is part of at least one edge in at least a
$1-\gamma$ fraction of the edge sets $E_i$.
Then, $G$ has a spanning tree of total cost strictly less than
$\frac{2}{1-2 \gamma} \cdot C$.
\end{lemma}

\begin{proof}
Let $v^i_e = 1$ if the edge set $E_i$ contains the edge $e$,
and $v^i_e = 0$ otherwise.
Let $v_e = \frac{1}{(1-2 \gamma) \cdot K} \cdot \sum_{i=1}^K v^i_e$.
We claim that the vector $\vc{v}$ satisfies the fractional Steiner
Tree LP
\[ \begin{array}{ll}
     \sum_{e \in (S, \Compl{S})} v_e \; \geq \; 1
     & \text{ for all cuts } (S, \Compl{S}),\\
     v_e \; \geq \; 0
     & \text{ for all edges } e.
\end{array} \]      
Consider any cut $(S, \Compl{S})$, and let $x \in S, y \notin S$ be
arbitrary.
Because at most a $\gamma$ fraction of the edge sets $E_i$
exclude each of $x$ and $y$,
at least a $1-2\gamma$ fraction of the $E_i$ include both $x$ and $y$.
By the connectivity assumption for the $E_i$,
each such $E_i$ must contain a path from $x$ to $y$,
and in particular an edge crossing $(S, \Compl{S})$.
Thus, we get that
\begin{align*}
  \sum_{e \in (S, \Compl{S})} v_e
  & = \frac{1}{(1-2\gamma) \cdot K} \cdot \sum_{i=1}^K \sum_{e \in (S, \Compl{S})} v^i_e 
    \; \geq \; \frac{1}{(1-2\gamma) \cdot K} \cdot (1-2\gamma) \cdot K
    \; = \; 1.
\end{align*}

The cost of this fractional solution is at most
\begin{align*}
\sum_{e=(x,y)} d(x,y) \cdot v_e
& = \sum_{e=(x,y)} d(x,y) \cdot \frac{1}{(1-2\gamma) \cdot K} \cdot \sum_{i=1}^K v^i_e
\\ & = \frac{1}{(1-2 \gamma) \cdot K} \cdot \sum_{i=1}^K \sum_{e=(x,y)} d(x,y) v^i_e
  \; < \; \frac{1}{(1-2\gamma) \cdot K} \cdot K \cdot C
  \; = \; \frac{C}{1-2\gamma}.
\end{align*}

Goemans and Williamson \cite{goemans:williamson:constrained-forest}
show that the integrality gap of the Steiner Forest LP is at most 2;
in particular, this means that there is a spanning tree of all the
nodes of total cost strictly less than
$\frac{2}{1-2\gamma} \cdot C$.
\end{proof}

\section{Main Result and Proof} \label{sec:main-result}
In this section, we prove the main result of our paper.

\begin{theorem} \label{thm:main-formal}
There is a polynomial-time algorithm (Algorithm~\ref{algo:log-approx})
which outputs a periodic tour $\hat{\SCH}$ simultaneously guaranteeing that
\begin{align*}
\MCost[2]{\hat{\SCH}} & \leq O(\log n) \cdot \MCost[2]{\SCH^*_{2}}, \\
\MCost[\infty]{\hat{\SCH}} & \leq O(\log n) \cdot \MCost[\infty]{\SCH^*_{\infty}}, \\
\end{align*}
where $\SCH^*_2$ and $\SCH^*_{\infty}$ are the optimal solutions for
the respective objective functions.

In other words, $\hat{\SCH}$ is simultaneously an $O(\log n)$
approximation for both of the weighted TSP objectives we consider.
\end{theorem}

First, by standard arguments, rounding all weights down to the nearest
power of 2 and then solving the resulting problem loses at most a
factor of 2 in the approximation guarantee.
We therefore assume from now on that all weights are powers of 2.
For each $i = 0, 1, \ldots, M$,
let \AP[i] be the set of all points of weight $2^{-i}$,
and $n_i = \SetCard{\AP[i]}$ the number of points of weight $2^{-i}$.
Let $M$ be largest such that at least one point has weight $2^{-M}$.

The high-level intuition behind our algorithm is the following:
in order to achieve a good objective value, each point should be
visited with frequency proportional to its weight.
However, those visits should also be ``roughly evenly spaced;''
many visits in short succession, followed by a long absence,
do not help either of the objectives.
To build such a tour systematically,
we consider individual tours for points of the same weight;
i.e., tours for points of weight $2^{-i}$, for each $i$.
These tours must then be sequenced carefully,
so that tours of high-weight points are more frequent.

More precisely, points of weight $2^{-i}$ are partitioned into $2^i$
sets, with a near-optimal tour through each set.
One \emph{phase} then consists of following one such tour for each $i$.
This means that each tour for a subset of points of weight $2^{-i}$
is followed roughly once every $2^i$ phases, balancing out the lower weight.
A straightforward implementation of this approach would give an
$O(M)$-approximation;
a more careful arrangement of phases
and subtours improves this bound to $O(\log n)$.

From the preceding discussion, it is clear that an important part of
the algorithm is the ability to partition all points of a given weight
$2^{-i}$ into $2^i$ sets so that the tour lengths through all of
these sets are comparable.
Thereto, a core subroutine is an approximation algorithm for
the Min-Max Spanning Tree Problem \cite{even:garg:konemann:ravi:sinha},
defined as follows:

\begin{problem}[Min-Max Spanning Tree]
Given an edge-weighted graph $G = (V,E)$ and a parameter $k$,
find $k$ edge sets $E_i \subseteq E$ such that each $E_i$ induces a
connected subgraph of $G$, and each node of $G$ is in at least one of
the induced subgraphs.
Subject to this constraint, minimize $\max_i c(E_i)$, where the cost
of an edge set is the sum of its edge weights.
\end{problem}

While this problem is NP-hard, the main result of
\cite{even:garg:konemann:ravi:sinha} provides a polynomial-time
constant-factor approximation algorithm:

\begin{theorem}[Theorem 4 of \cite{even:garg:konemann:ravi:sinha}]
There is a polynomial-time 4-approximation\footnote{Technically, the
  result of Even et al.~guarantees a $(4+\epsilon)$-approximation for
  every $\epsilon$. We ignore the $\epsilon$ term here.}
algorithm for the Min-Max Spanning Tree problem.
In other words, if $E^*_1, E^*_2, \ldots, E^*_k$ is the optimal
solution and $E_1, E_2, \ldots, E_k$ the solution returned by the
algorithm,
\begin{align*}
  \max_i c(E_i) & \leq 4 \cdot \max_i c(E^*_i).
\end{align*}
\end{theorem}

Our algorithm is given formally as Algorithm~\ref{algo:log-approx}.

\begin{algorithm}[htb]
\begin{algorithmic}[1]
\FORALL{pairs $(x,y)$ of points}
  \STATE Let the cost of $(x,y)$ be the metric space distance $d(x,y)$.
\ENDFOR
\FOR{$i = 0$ \TO $M$}
\STATE Let $\theta_i = \min(n_i, 2^i)$.
\STATE Use the 4-approximation algorithm of
\cite{even:garg:konemann:ravi:sinha} (Algorithm~\ref{algo:even})
to compute
$\theta_i$ disjoint trees $T^i_1, T^i_2, \ldots, T^i_{\theta_i}$
such that each point of \AP[i] is in one of the $T^i_j$.
\FORALL{$j$}
  \STATE Shortcut the Eulerian tour of $T^i_j$ into a tour $\SCH^i_j$.
\ENDFOR
\ENDFOR
\STATE Let $\SCH_1, \SCH_2, \ldots, \SCH_J$ be an enumeration of
all (at most $n$) tours $\SCH^i_j$ by non-increasing $i$.\\
\COMMENT{If $i' > i$, then $\SCH^{i'}_{j'}$ must precede $\SCH^i_j$.
The ordering of tours with the same $i$ is arbitrary.}
\STATE Let $I = \ceil{\log_2 (J+1)} - 1$.
\FOR{$i = 0, \ldots, I-1$}
\STATE Let $L_i$ be the list $[\SCH_{2^i}, \SCH_{2^i+1}, \ldots, \SCH_{2^{i+1}-1}]$,
and $\lambda_i = 2^i$ the number of tours in $L_i$.
\ENDFOR
\STATE Let $L_I$ be the list $[\SCH_{2^I}, \SCH_{2^I+1}, \ldots, \SCH_{J}]$,
and $\lambda_I = J + 1 - 2^I$ the number of tours in $L_I$.
\STATE Output the following schedule $\hat{\SCH}$:
\FOR{each phase $j = 0, 1, 2, \ldots$} \label{line:outer-loop}
    \FOR{$i = 0$ \TO $I$}
        \STATE Execute the entire tour number $(j \mod \lambda_i)$ of $L_i$.
    \ENDFOR
\ENDFOR
\end{algorithmic}
\caption{The $O(\log n)$ approximation algorithm. \label{algo:log-approx}}
\end{algorithm}


The tours of high-weight points end up in lists $L_i$ with small values
of $i$.
They are toured exponentially more frequently than tours in lists
$L_i$ with large values of $i$.
A good example case to keep in mind in the following analysis is when
there are more than $2^i$ points of weight $2^{-i}$ --- then, each list $L_i$
contains exactly the $2^i$ tours of points of weight $2^{-i}$.
The more elaborate construction of lists $L_i$ is needed to
achieve the $O(\log n)$ guarantee.
The guarantee on the relationship between $L_i$ and points' weights is
captured by the following lemma:

\begin{lemma} \label{lem:list-for-weight}
For each $i$, each tour $\SCH^i_j$ in the list $L_i$,
and any point $x$ occurring in $\SCH^i_j$,
the weight of $x$ is at most $w_x \leq 2^{-i}$.
\end{lemma}

\begin{proof}
There is at most one tour $\SCH_1$ of points of weight $1$,
at most two tours of points of weight $\frac{1}{2}$;
more generally, there are $\theta_i \leq 2^i$ tours of points of weight $2^{-i}$.
Therefore, points of weight $2^{-i}$ must appear in one of the first
$2^i-1$ tours, and therefore be in $L_{i'}$ for $i' \leq i$.
\end{proof}

A key lemma for the analysis is the following lower bound on the
optimum MST cost for points of weight $2^{-i}$.
Recall here and below that the costs are the metric distances,
and for trees $T$ or tours \SCH,
we write $d(T)$ (or $d(\SCH)$) for the sum of costs of all edges in
$T$ or \SCH.

\begin{lemma} \label{lem:key-lower-bound}
For any $i$, let \OptTreeCost[i] be the cost of the minimum spanning
tree of \AP[i].
Then, for all $j$, we have that
$\theta_i \cdot d(T^i_j) \leq 4 \OptTreeCost[i]$.
\end{lemma}

\begin{proof}
The proof of this lemma is directly based on the analysis of
\cite{even:garg:konemann:ravi:sinha}.
We need a slightly more detailed structural analysis of the algorithm
than expressed in the lemmas and theorems in
\cite{even:garg:konemann:ravi:sinha},
so we restate and analyze the algorithm here as Algorithm~\ref{algo:even}.
In the algorithm, $B$ is a guess for the optimum Min-Max Spanning Tree
objective value \OPTMINMAX.
When the algorithm terminates successfully,
it returns a solution of objective value at most $4B$,
while guaranteeing that $B \leq \OPTMINMAX$.
The analysis in \cite{even:garg:konemann:ravi:sinha} shows that a
successful value of $B$ is always found.

\begin{algorithm}[htb]
\begin{algorithmic}[1]
\STATE Remove all ``heavy'' edges of weight greater than $B$.\\
Let $\{G_i\}_i$ denote the connected components after deleting the heavy edges.
\STATE Let $\MST_i \gets \text{ minimum spanning tree of } G_i$.
\STATE Let $k_i \gets \floor{\frac{d(\MST_i)}{2B}}$.
\IF{$\sum_{i}(k_i + 1) > k$}
\RETURN{``Fail --- $B$ is too low''.}
\ELSE
\STATE Edge-decompose each tree $\MST_i$ into at most $(k_i+1)$ trees
$T^i_0, T^i_1, \ldots, T^i_{k_i}$
such that $d(T^i_0) < 2B$ and  $d(T^i_j) \in [2B, 4B)$ for all $j > 0$.
\RETURN{``Success --- set of trees $\{S^i_j\}_{i,j}$''.}
\ENDIF
\end{algorithmic}
\caption{The Min-Max Spanning Tree 4-approximation algorithm of
  \cite{even:garg:konemann:ravi:sinha}. \label{algo:even}}
\end{algorithm}

Consider a successful run of Algorithm~\ref{algo:even}.
Let $c$ be the number of connected components $G_i$ after removing the
heavy edges.
Because the MST of all points is connected, it must contain at least
$c-1$ edges of weight at least $B$.
Only the trees $T^i_0$ can have weight less than $2B$,
while all other trees $T^i_j$ have weight at least $2B$.
By the execution of Kruskal's Algorithm, each $\MST_i$ is a subtree of
the overall MST, and because each $\MST_i$ is broken into pieces to
produce the $T^i_j$, each $T^i_j$ is a subtree of the overall MST as
well.
Therefore, the cost of the overall MST is at least
$(k-c) \cdot 2B + (c-1) B$.
For $c < k$, this quantity is obviously lower-bounded by
$kB$.
When $c = k$, notice that at least one tree $T^i_j$ must have total
edge cost at least $\OPTMINMAX \geq B$ (since the solution is feasible),
so the overall cost of the MST is at least
$(k-1) B + \OPTMINMAX \geq kB$.

Because the maximum cost of any tree $T^i_j$ is at most $4B$,
we obtain the claimed bound.
\end{proof}




We are now ready to prove Theorem~\ref{thm:main-formal}.

\begin{extraproof}{Theorem~\ref{thm:main-formal}}
Notice that for every tour \SCH, we have that
$\MCost[2]{\SCH} \leq \MCost[\infty]{\SCH}$, because any convex
combination of absence lengths is upper-bounded by the maximum absence
length.
We will show the stronger statement that 
$\MCost[\infty]{\hat{\SCH}} \leq O(\log n) \cdot \MCost[2]{\SCH^*_{2}}$.
This implies the first part of the theorem because
$\MCost[2]{\hat{\SCH}} 
\leq \MCost[\infty]{\hat{\SCH}}
\leq O(\log n) \cdot \MCost[2]{\SCH^*_{2}}$,
and the second part of the theorem because
$\MCost[\infty]{\hat{\SCH}} \leq O(\log n) \cdot \MCost[2]{\SCH^*_{2}}
\leq O(\log n) \cdot \MCost[2]{\SCH^*_{\infty}}
\leq O(\log n) \cdot \MCost[\infty]{\SCH^*_{\infty}}$.

Write $\OPTVALUE = \MCost[2]{\SCH^*_2}$.
Consider an arbitrary point $\hat{x}$,
of weight $2^{-\hat{\imath}}$.
We want to upper-bound the absence length under $\hat{\SCH}$
between any two consecutive visits to $\hat{x}$.
By Lemma~\ref{lem:list-for-weight}, $\hat{x}$ must appear in a tour in
a list $L_{i'}$ for $i' \leq \hat{\imath}$.
In Line~\ref{line:outer-loop}, $\lambda_{i'}$
phases elapse before a return to $\hat{x}$.
Each of those phases traverses exactly one tour from each list
$L_{i}$, for $i = 0, \ldots, I$.
In addition, it requires moving from each tour to the next.

The tours $\SCH_j$ are disjoint,
and are themselves partitioned disjointly among the lists $L_i$.
Let $U_i$ be the set of points $x$ such that a tour in $L_i$ visits $x$.
The $U_i$ form a disjoint partition of \AP
(similar, but typically not identical, to the \AP[i]). 
We let $d(U_i, U_j) = \max_{x \in U_i, y \in U_j} d(x,y)$
be the maximum distance between any pair of points from $U_i, U_j$,
and $\hat{d}(L_i) = \max_{\SCH_j \in L_i} d(\SCH_j)$ the maximum total length
of any tour in $L_i$.

Then, the preceding argument implies that the total length of an
absence from $\hat{x}$ under $\hat{\SCH}$ is upper-bounded by
\begin{equation}
\lambda_{i'} \cdot
\sum_{i=0}^I \left( \hat{d}(L_i) + d(U_i, U_{(i+1) \mod I}) \right).
\label{eqn:objective-bound}
\end{equation}

We bound both terms under the sum separately.
First, we upper-bound
$d(U_i, U_{(i+1) \mod I}) \leq \max_{x,y \in \AP} d(x,y)$;
next, we prove that $d(x,y) \leq \OPTVALUE$.
Let $z$ be a point of weight $1$.
Because the optimum tour must at some point visit $x$, there is an
absence interval from $z$ of length at least $2d(x,z)$.
Similarly for $y$.
Therefore,
\[ \OPTVALUE \; \geq \; 1 \cdot \max(2d(x,z), 2d(y,z))
  \; \geq \; d(x,z) + d(y,z)
  \; \geq \; d(x,y).
\]

Next, we bound $\hat{d}(L_i) \leq O(\OPTVALUE)$.
We distinguish two cases.
If $\theta_i = n_i$, then each point of weight $2^{-i}$ is in its own
spanning tree and tour,
so $d(\SCH^i_j) = 0$ for all $j$.
In particular $\hat{d}(L_i) = 0$, and the claim holds trivially.

Otherwise, $\theta_i = 2^i$, and Lemma~\ref{lem:key-lower-bound}
implies that $d(T^i_j) \leq 4 \cdot 2^{-i} \cdot \OptTreeCost[i]$.
By the standard shortcutting argument for the Euclidean tour of a tree,
$d(\SCH^i_j) \leq 2 d (T^i_j)$.
We now show that $2^{-i} \cdot \OptTreeCost[i] \leq O(\OPTVALUE)$.
Consider the subtour $\SCH'$ of $\SCH^*_2$ induced only by points of
weight $2^{-i}$, skipping all other points.
For any point $x$ of weight $2^{-i}$, the objective value is
$2^{-i} \PCost[2]{x}{\SCH'} \leq 2^{-i} \PCost[2]{x}{\SCH^*_2} \leq \OPTVALUE$.
By Lemma~\ref{lem:quadratic-vs-TSP},
$\max_{x \in \AP[i]} \PCost[2]{x}{\SCH'} \geq \Omega(\OptTreeCost[i])$,
implying the claimed bound.

Substituting all these bounds into \eqref{eqn:objective-bound},
we can now upper-bound the length of absence from $\hat{x}$ by
$\lambda_{i'} \cdot (I+1) \cdot O(\OPTVALUE)$.

Because $\lambda_{i'} \leq 2^{i'}$ and the weight of $\hat{x}$ is
$2^{-\hat{\imath}} \leq 2^{-i'}$,
the weighted objective function value for $\hat{x}$ is at
most $O((I+1) \OPTVALUE)$.
There are $J \leq n$ tours,
so $I = \ceil{\log_2 (J + 1)} - 1 \leq \log_2 n$.
Finally, we lost only a factor of 2 due to rounding the weights to
powers of 2. 
Because $\hat{x}$ was chosen arbitrarily,
we have shown that Algorithm~\ref{algo:log-approx} is 
an $O(\ceil{\log_2 n}) = O(\log n)$ approximation.
\end{extraproof}

\begin{remark}
\begin{itemize}
\item In the proof, we bounded $I$ by $\log_2 n$.
The analysis also directly gives a bound of $O(M)$.
When the range of point weights is small (i.e., $M$ is small),
this bound may be better; in other words, the algorithm exploits
favorable point weights.

\item The constant in big-$O$ we obtain is large.
A more reasonable bound of $17 \log_2 n$ can be obtained for the
objective \MCost[\infty]{\hat{\SCH}}, by comparing the tour $\hat{\SCH}$
directly to $\SCH^*_{\infty}$ instead of $\SCH^*_2$,
and using Lemma~\ref{lem:TSP-to-weighted} instead of
Lemma~\ref{lem:quadratic-vs-TSP}.
The constant can be further improved to 13 by
using the improved 3-approximation algorithm for Min-Max Spanning Tree 
due to \cite{khani:salavatipour:min-max}
instead of the 4-approximation of \cite{even:garg:konemann:ravi:sinha}.
\end{itemize}
\end{remark}

\section{A Reduction from Security Games} \label{sec:additional-results}
We show that the security game mentioned in
Section~\ref{sec:security-games-intro} can be reduced to the weighted
quadratic TSP objective, at a loss of a constant factor in the
objective function.
In the security game of \cite{QuasiRegular},
there are $n$ targets with weights $w_x$.%
\footnote{In \cite{QuasiRegular}, all targets are at unit distance
  from each other. Our main goal here is to leverage our approximation
  algorithm for the weighted Quadratic TSP objective to obtain (weaker)
  guarantees for the version of the security game in which there is an
  arbitrary metric on the targets.} 
A defender commits to a distribution over \emph{schedules}
$\sigma = (\SCH, \tau)$,
consisting of a tour \SCH and a real-valued time offset $\tau$.
The intuition is that the defender will execute the tour \SCH,
but offset/delayed by $\tau$, to create uncertainty.%
\footnote{To be precise, the model of \cite{QuasiRegular} is entirely
defined in continuous time, as a mapping 
$\sigma: \mathbb{R}^{\geq 0} \to \AP \cup \SET{\perp}$,
with $\perp$ denoting that the defender is in transit between targets.
Such schedules allow staying at points for a positive amount of time,
something that we do not allow in this paper.}
The attacker then chooses one triple $(x, t_0, t)$;
here, $x$ is a target, $t_0$ the start time of the attack,
and $t$ its duration.
The attack $(x, t_0, t)$ is \emph{successful} if the defender does not
visit target $x$ in the interval $[t_0, t_0 + t]$ under $\sigma$.
The defender's cost (and the attacker's utility) is
\begin{align*}
  U(\sigma, (x, t_0, t)) & =
  \begin{cases}
   w_{x} \cdot t & \text{ if the attack is successful},\\
0 & \text{ otherwise}.
  \end{cases}
\end{align*}
The defender chooses a mixed strategy
(distribution $\Sigma$) over schedules $\sigma$
so as to minimize the attacker's expected utility.

Lemma B.3 of \cite{QuasiRegular} shows that without loss of
generality, the optimal mixed defender strategy $\Sigma^*$
is \emph{shift-invariant} in the following sense:
for each $x, t$, the distribution of return times to target $x$
starting from time 0 is the same as starting from time $t$.

\begin{lemma}[Lemma B.3 of \cite{QuasiRegular}] \label{lem:shift-invariant}
There is a shift-invariant optimal mixed defender strategy.
\end{lemma}

Lemma~\ref{lem:shift-invariant} allows us to assume w.l.o.g.~that the
attacker chooses $t_0 = 0$ (or any other fixed time),
since no time is better for him than any other.
This will be useful in the subsequent analysis.

We now lower-bound the objective function in terms of 
the cost functions \PCost[2]{x}{\SCH} for targets $x$.
The first step is to use Proposition 3.2 of \cite{QuasiRegular},
which does not rely on uniform distances between targets.

\begin{lemma}[Proposition 3.2 of \cite{QuasiRegular}]
\label{lem:attacker-lower-bound}
Consider any shift-invariant strategy $\Sigma$ and target $x$,
and let $T_x$ be the expected time until the next return of the defender
to target $x$ after time $0$ (under $\Sigma$).
By choosing an attack duration of
$t = \frac{T_x}{2}$,
the attacker can achieve expected utility at least
$w_x \cdot \frac{T_x}{4}$.
\end{lemma}

Our next goal is to lower-bound $T_x$.
Suppose that the attacker draws his start time $t_0$ uniformly at
random from the interval $[0, \hat{t}]$, for some value of $\hat{t}$.
Consider the distribution of the defender's next return time to
$x$ after time $t_0$.
By shift-invariance,
this distribution is the same for all possible draws $t_0$. 
We lower-bound the expected return time by pretending that the
defender returns to all targets at time $\hat{t}$.
Consider a fixed (i.e., non-random) defender schedule
$\sigma = (\SCH, \tau)$.
The expected return time to target $x$ under $\sigma$
for a uniformly random attack start time $t_0$
is exactly half the sum of squares of interval lengths between
consecutive visits to $x$ over $[0, \hat{t}]$.
Denote this quantity by $s(\sigma, x, \hat{t})$.
We can lower-bound
$T_x \geq \Expect[\sigma = (\SCH, \tau) \sim \Sigma]{%
  \frac{1}{\hat{t}} \cdot s(\sigma, x, \hat{t})}$.
By taking the limit of $\hat{t} \to \infty$,
and noticing that
$\lim_{\hat{t} \to \infty} \frac{1}{\hat{t}} \cdot s(\sigma, x, \hat{t}) =
\PCost[2]{x}{\SCH}$, we get that
$T_x \geq \Expect[\sigma = (\SCH, \tau) \sim \Sigma]{\PCost[2]{x}{\SCH}}$.
Overall, we have shown that the attacker's expected utility
against a mixed defender strategy $\Sigma$ is at least
$\Omega\left(\max_x
  \Expect[\sigma = (\SCH, \tau) \sim \Sigma]{w_x \PCost[2]{x}{\SCH}}\right)$.

Below, we prove Lemma~\ref{lem:derandomize},
which shows that we can approximate this objective with a single
periodic tour and a uniformly random offset.

\begin{lemma} \label{lem:derandomize}
For every mixed defender strategy $\Sigma$,
there is a single periodic tour \SCH with
$\MCost[2]{\SCH} \leq O(1) \cdot
\max_x \left(\Expect[\sigma = (\SCH, \tau) \sim \Sigma]{w_x \PCost[2]{x}{\SCH}}\right)$.
\end{lemma}

We can now finish the reduction by upper-bounding the defender's cost
under the periodic tour from Lemma~\ref{lem:derandomize}.

\begin{lemma} \label{lem:attacker-upper-bound}
For any periodic tour \SCH with uniformly random offset $\tau$,
by attacking target $x$,
the attacker can obtain expected utility at most 
$w_x \cdot \frac{1}{2} \cdot \PCost[2]{x}{\SCH}$. 
\end{lemma}

\begin{proof}
Because the offset $\tau$ is uniformly random,
the strategy is shift-invariant,
and the attacker attacks at time 0 without loss of generality.
Even if he exactly knew when after time 0 the defender would return
next, the attacker could not in expectation attack longer than the
expected time until the defender's next return.
For a periodic schedule \SCH, the expected time until the defender's
next return is $\frac{1}{2} \cdot \PCost[2]{x}{\SCH}$.
\end{proof}

Thus, we have shown that the objective value of the 
optimum finite periodic tour \SCH is within a constant factor of the
optimum randomized defender strategy.
Therefore, any approximation guarantee for the \MCost[2]{\SCH}
objective for finite periodic tours carries over --- at the loss of a
constant factor --- to the security game objective from
\cite{QuasiRegular}.
In particular, using the result from Theorem~\ref{thm:main-formal},
we obtain a polynomial-time $O(\log n)$ approximation algorithm for
finding a defender strategy minimizing a rational attacker's expected
utility.

\begin{extraproof}{Lemma~\ref{lem:derandomize}}
By Lemma~\ref{lem:attacker-lower-bound} and the subsequent discussion,
the attacker's utility from attacking $x$
against a mixed defender strategy $\Sigma$ is
$\Omega\left(\Expect[\sigma = (\SCH, \tau) \sim \Sigma]{w_x \PCost[2]{x}{\SCH}}\right)$.
By Lemma~\ref{lem:periodic-is-near-OPT},
at a cost of a factor 3 in the \PCost[2]{x}{\SCH},
we can ensure that all tours \SCH in the support of $\Sigma$ are
periodic.
Then, $\Sigma$ is of the following form:
(1) Choose a (finite, periodic) tour \SCH from some distribution.
Let $D$ be the total length of executing the finite tour once.
(2) Shift the starting point of the tour to a uniformly random point
in the (continuous) interval $\tau \in [0,D]$.

Because there are only countably many periodic tours \SCH,
the distribution in step (1) must be discrete.
For each tour \SCH, let $p_{\SCH} \geq 0$ be the probability of
choosing \SCH.
Sort the tours \SCH such that $p_{\SCH_1} \geq p_{\SCH_2} \geq \cdots$;
for ease of notation, write $p_i = p_{\SCH_i}$.
Because $\lim_{k \to \infty} \sum_{i=1}^k p_i = 1$, there must
be a finite $k$ such that $\sum_{i=1}^k p_i \geq \frac{1}{2}$;
fix the smallest such $k$, and let $q = p_k$.

For each of the $\SCH_i$, let $D_i$ be the total length of one
traversal of the tour.
Let $\bar{D} = \max_{i=1, \ldots, k} D_i$ be the length of the longest of the
tours.
Let $K = \max_{i,x} \frac{8 \bar{D}}{\PCost[2]{x}{\SCH_i}}$.
Let $N_i := \ceil{K \cdot \frac{p_i}{q} \cdot \frac{\bar{D}}{D_i}}$.

Consider the finite, periodic tour $\hat{\SCH}$
that first performs $N_1$ iterations  of  $\SCH_1$, 
then $N_2$ iterations of  $\SCH_2$, 
then $N_3$ iterations of  $\SCH_3$,
and so on until $N_k$ iterations of $\SCH_k$.
(Rotate each tour so that the start point of the next tour is the end
point of the previous tour.)
Let $\hat{D} = \sum_{i=1}^k N_i \cdot D_i$ be the total length
of this tour.
Randomize by choosing a uniformly random start time
from $[0, \hat{D}]$.
We claim that this strategy has objective value within a constant
factor of that of $\Sigma$.

Fix any point $x$.
Because we are simply leaving out some terms from the sum (without
renormalizing), we get that
$\Expect[\SCH \sim \Sigma]{\PCost[2]{x}{\SCH}}
\geq \sum_{i=1}^k p_i \PCost[2]{x}{\SCH_i}$.
Now consider \PCost[2]{x}{\hat{\SCH}}.
It is exactly the expected length of the absence interval from $x$
for a uniformly randomly chosen time $t \in [0, \hat{D}]$.

The total time of $\hat{\SCH}$ that is devoted to tours $\SCH_i$ for a
given $i$ is $N_i D_i \leq 2K \cdot \frac{p_i}{q} \cdot \bar{D}$.
Conditioned on the random time $t$ falling into an interval devoted to
$\SCH_i$, the expected absence length from $x$ is almost exactly
\PCost[2]{x}{\SCH_i}.
The only problems could arise in the first or last iteration, where
the next (or previous) visit to $x$ may belong to a different tour,
and have larger expectation.
In that case, we upper-bound the absence length by $2\bar{D}$.
Therefore, the expected absence length is at most

\begin{align*}
\frac{(N_i-2) D_i \PCost[2]{x}{\SCH_i} + 2 (2\bar{D})^2}{N_i \cdot D_i}
  & \leq \PCost[2]{x}{\SCH_i} + \frac{8 \bar{D}^2}{N_i \cdot D_i}
  \; \leq \; \PCost[2]{x}{\SCH_i} + \frac{q \PCost[2]{x}{\SCH_i}}{p_i}
  \; \leq \; 2 \PCost[2]{x}{\SCH_i}.
\end{align*}

The probability that the random time $t$ falls into an interval
devoted to $\SCH_i$ is
\begin{align*}
\frac{N_i D_i}{\sum_j N_j D_j}
  & \leq \frac{2K \frac{p_i}{q} \cdot \bar{D}}{%
    \sum_j K \frac{p_j}{q} \cdot \bar{D}}
  \; = \; \frac{2 p_i}{\sum_j p_j}
  \; \leq \; 4 p_i,
\end{align*}
because the sum of probabilities $p_j$ is at least \half.
In summary, the expected absence length from $x$ is at most
$\sum_i (4p_i) \cdot (2 \PCost[2]{x}{\SCH_i})
\leq 8 \Expect[(\SCH, D) \sim \Sigma]{\PCost[2]{x}{\SCH}}$.
\end{extraproof}

\section{Conclusion} \label{sec:conclusions}
Motivated in part by applications in security games such as
anti-burglary or anti-poaching patrols,
as well as an economic view of the traveling salesman problem,
we defined a class of weighted traveling
salesman problems in which the goal is to visit the points of a
metric space repeatedly,
while minimizing the weighted duration of absences from the points.
We gave a combinatorial algorithm that simultaneously achieves an
$O(\log n)$ approximation guarantee for two natural objective
functions.


One can generalize the cost functions for individual points to
\[\PCost[p]{x}{\SCH}
  = \lim_{k \to \infty}
  \frac{\sum_{k'=0}^{k} (d_{\SCH}(\VSch[x]{k'},\VSch[x]{k'+1}))^p}{%
    \sum_{k'=0}^{k} (d_{\SCH}(\VSch[x]{k'},\VSch[x]{k'+1}))^{p-1}}.\]
Then, the two special cases we analyzed are $p=2$ and $p\to\infty$.
Because for any fixed $x, \SCH$, the cost function
\PCost[p]{x}{\SCH} is monotone non-decreasing in $p$,
so is \MCost[p]{\SCH}.
We showed in Section~\ref{sec:main-result} that the algorithm's
tour $\hat{\SCH}$ satisfies
$\MCost[\infty]{\hat{\SCH}} \leq O(\log n) \cdot \MCost[2]{\SCH^*_2}$.
This immediately implies that Algorithm~\ref{algo:log-approx} is an
$O(\log n)$ approximation for \MCost[p]{\SCH} for all $p \geq 2$.
However, we are not aware of any natural motivation for studying
the objective for $p \in (2, \infty)$.
For $p < 2$, the objective function suffers, in special cases,
from rewarding repeated visits to the same point (without traveling
any distance in between), which is why we did not study it.

The most obvious open question left by our work is whether the
approximation guarantee can be improved to a constant.
For many special cases, natural modifications of 
Algorithm~\ref{algo:log-approx} do give constant-factor
approximations. It seems plausible that a more careful interleaving of
subtours could lead to such an improved guarantee.

Instead of minimizing the \emph{maximum} of the weighted point
objectives, one could study a weighted \emph{sum}.
Such an objective would be motivated if a need for the salesman's
product arises with a uniformly randomly chosen person,
and we want to minimize the expected time to serve the demand.
First notice that in this version, the weights can be assumed to be
uniform, as a point of weight $w$ can be replaced by $w$ points
of weight 1, without any change in the objective.
The resulting problem becomes more similar to the well-known Traveling
Repairman (a.k.a.~Minimum Latency) problem \cite{arora:karakostas:latency,BCCPRS:minimum-latency,goemans:kleinberg:latency}.
However, even for the (unweighted) sum version of our problem,
the standard techniques for the Minimum Latency problem do not appear
to carry over, 
for reasons similar to those discussed in Section~\ref{sec:related-work}.
Algorithm~\ref{algo:log-approx} in the case of no weights simply
returns a TSP tour of all points.
It is not difficult to construct examples where this can be
polynomially far from optimal.
Hence, studying the sum (instead of maximum) version of the problem is
an interesting direction for future work.

Another natural question, motivated by the security game application,
is whether one can also obtain good approximation guarantees for the
same objective in the presence of multiple ``salesmen;''
for example, national parks will typically have multiple anti-poaching
patrols simultaneously.
The work on Min-Max Tree Covers
\cite{even:garg:konemann:ravi:sinha,khani:salavatipour:min-max}
we used as a key tool was motivated by a similar goal of dividing
points between multiple TSP tours.
See also \cite{arkin:hassin:levin:vehicle,ForestMultirobot},
and \cite{fakcharoenphol:harrelson:rao:repairman} for a discussion of
the $k$-Traveling Repairman Problem.

We studied two natural cost minimization versions of the weighted TSP.
In recent work, motivated by some
of the same application domains (such as poaching prevention),
Immorlica and Kleinberg \cite{immorlica:kleinberg:bandits} defined a
natural prize-collecting problem with a time component.
In their problem, distances between ``points'' are uniform.
Each point $x$ has an increasing concave reward function $f_x$.
When a point is visited after an absence interval of length $t$, 
it yields a reward of $f_x(t)$.
This objective also naturally motivates visiting more valuable points
more frequently, and spacing out visits roughly evenly.
It is an interesting question whether the techniques in our paper can
be adapted to obtain provable approximation guarantees for the
prize-collecting version of \cite{immorlica:kleinberg:bandits} under
non-uniform distances.
One possible difficulty is that while in our cost-minimization
formulation, each point must be visited infinitely often to prevent
infinite cost, the prize-collecting version can simply omit some
points if they are far away and/or provide little reward.

\subsubsection*{Acknowledgments}
We would like to thank Shaddin Dughmi, Bobby Kleinberg, Kamesh
Munagala, Debmalya Panigrahi, David Shmoys, and Omer Tamuz for useful
discussions, and anonymous reviewers for useful feedback.

\bibliographystyle{plain}
\bibliography{davids-bibliography/names,davids-bibliography/conferences,davids-bibliography/bibliography,davids-bibliography/publications,bibliography/paper-specific}

\end{document}